\documentclass{llncs}

\usepackage{microtype}

\title{Computing Euclidean k-Center over Sliding Windows\thanks{This work has been supported by DFG grant Kl 655/19 as part of a DACH project}}

\author
{Sang-Sub Kim}
\institute{Department of Computer Science, University of Bonn, Germany\\
\email{sangsub@uni-bonn.de    }}

\authorrunning{S. Kim}

\usepackage[normalem]{ulem}
\usepackage[utf8]{inputenc}

\usepackage{mathtools} 
\usepackage{wrapfig} 

\usepackage{amsmath}
\usepackage{algorithm}
\usepackage[noend]{algpseudocode}
\usepackage{multicol}

\usepackage{xcolor}
\usepackage{xspace}
\usepackage{fmtcount} 
\usepackage{xstring}
\usepackage{cite}

\usepackage{nicefrac}

\usepackage{bookmark}	

\newtheorem{invariant}{Invariant}

\newcommand{\rst}{r^*}
\newcommand{\eps}{\epsilon}

\usepackage{caption} 

\begin{document}
\maketitle

\begin{abstract}

In the Euclidean $k$-center problem in sliding window model, input points are given in a data stream and the goal is to find the $k$ smallest congruent balls whose union covers the $N$ most recent points of the stream. In this model, input points are allowed to be examined only once and the amount of space that can be used to store relative information is limited.

Cohen-Addad et al.~\cite{cohen-2016} gave a $(6+\epsilon)$-approximation for the metric $k$-center problem using O($k/\epsilon \log \alpha$) points, where $\alpha$ is the ratio of the largest and smallest distance and is assumed to be known in advance.

In this paper, we present a $(3+\eps)$-approximation algorithm for the Euclidean $1$-center problem using O($1/\eps \log \alpha$) points. We present an algorithm for the Euclidean $k$-center problem that maintains a coreset of size $O(k)$.
Our algorithm gives a $(c+2\sqrt{3} + \epsilon)$-approximation for the Euclidean $k$-center problem using O($k/\epsilon \log \alpha$) points by using any given $c$-approximation for the coreset where $c$ is a positive real number.
For example, by using the $2$-approximation~\cite{feder-greene-1988} of the coreset, our algorithm gives a $(2+2\sqrt{3} + \eps)$-approximation ($\approx 5.465$) using $O(k\log k)$ time.
This is an improvement over the approximation factor of $(6+\epsilon)$ by Cohen-Addad et al.~\cite{cohen-2016} with the same space complexity and smaller update time per point.
Moreover we remove the assumption that $\alpha$ is known in advance. Our idea can be adapted to the metric diameter problem and the metric $k$-center problem to remove the assumption. 
For low dimensional Euclidean space, we give an approximation algorithm that guarantees an even better approximation.

\end{abstract}

\section{Introduction and problem statement}

The k-center problem, which is finding the $k$ smallest congruent balls containing a set of input points, is a fundamental problem arising from abundant real-world applications, including machine learning, data mining, and image processing. The growth of the Internet and the computing power of machines has facilitated a significant increase in the amount of data collected and used by various applications over the last decades. But in huge data sets it is quite hard to guarantee reasonable processing time and memory space. To cope with this difficulty, data stream models have received considerable attention in the theoretical as well as the application field.

In the streaming model, it is important to design an algorithm whose space complexity does not depend on the size of the input, since the memory size is typically much smaller than the input size. In this paper, we consider the single-pass streaming model~\cite{chan-pathak-2014}, where elements in the data stream are allowed to be examined only once and only a limited amount of information can be stored. The insertion-only stream model is well studied, but also more flexible settings like dynamic streams and the sliding window model have received some attention for many clustering problems. In the dynamic stream model input points can be removed arbitrarily and in the sliding window model older input is deleted as new elements arrive.


In this paper, we consider the Euclidean $k$-center problem for a sliding window, which contains the most recent $N$ points.

\subsection{Previous Work}

\textbf{The $k$-center problem in the static setting}.
The Euclidean $k$-center problem has been extensively studied in the literature.
If $k$ is part of the input, the $k$-center problem is $\mathrm{NP}$-hard~\cite{garey-johnson-1979}, even in the plane~\cite{megiddo-supowit-1984}.
In fact, it is known to be $\mathrm{NP}$-hard to approximate the $k$-center problem with a factor smaller than $2$ for arbitrary metric spaces~\cite{fowler-paterson-1981}, and with a factor smaller than $1.822$ for the Euclidean space~\cite{bern-eppstein-1996}.
If the Euclidean dimension $d$ is not fixed, the problem is $\mathrm{NP}$-hard for fixed $k \ge 2$~\cite{megiddo-1990}. Agarwal and Procopic~\cite{agarwal-procopiuc-2002} gave an exact algorithm that runs in O$(n^{k^{1-(1/d)}}$) for the $L_2$-metric and the $L_{\infty}$-metric.
Feder and Greene~\cite{feder-greene-1988} gave a $2$-approximation that runs in O($n \log k$) for any $L_p$-metric.









For small $k$ and $d$, better solutions are available. The $1$-center problem in fixed dimensions is known to be of LP-type and can be solved in linear time~\cite{chazelle-matousek-2015}.
For the Euclidean $2$-center problem in the plane, the best known algorithm is given by Chan~\cite{chan-1999}, which runs deterministically in O($n \log^2 n \log^2 \log n$) time using O($n$) space. \\
\textbf{The $k$-center problem in the streaming model}.
In the streaming model, where only a single pass over the input is allowed, McCutchen and Khuller~\cite{mccutchen-khuller-2008}, and independently Guha~\cite{guha-2009}, presented algorithms to maintain a $(2 + \eps)$-approximation for $k$-centers in any dimension using O($(kd/\eps) \log(1/\eps)$) space and O($(kd/\eps) \log(1/\eps)$) update time. For $k = 1$, Zarrabi-Zadeh and Chan~\cite{zarrabi-chan-2006} presented a simple algorithm achieving an approximation factor of $3/2$ using only O($d$) space. Agarwal and Sharathkumar~\cite{agarwal-shara-2010} improved the approximation factor to $(1 + \sqrt{3})/2 + \eps \approx 1.37$ using O($(d(1/\eps)^3) \log(1/\eps)$) space and O($(d(1/\eps)^2) \log(1/\eps)$) update time. The approximation factor of their algorithm was later improved to $1.22$ by Chan and Pathak~\cite{chan-pathak-2014}. For $k = 2$, Kim and Ahn~\cite{kim-ahn-2015} gave a $(1.8 + \eps)$-approximation using O($d/\eps$) space and O($d/\eps$) update time. Their algorithm extends to any fixed $k$ with the same approximation factor.
If $d$ and $k$ are fixed, Zarrabi-Zadeh~\cite{zarrabi-2008} gave an $(1+\eps)$-approximation using O($k(1/\eps)^d$) space and O($k(1/\eps)^d$) update time. \\
In the sliding window model, Cohen-Addad et al.~\cite{cohen-2016} have recently obtained results for a variant of the $k$-center problem that returns $k$ centers, but not the radius.
Under the assumption that the algorithm cannot create new points, they show that any randomized sliding window algorithm achieving
an approximation factor less than 4 with probability greater than $1/2$
for metric $2$-center problem requires $\Omega(\sqrt[3]{N})$ points.
For $k=2$, they give a $(4+\eps)$-approximation using O($1/\eps \log \alpha)$ points and O($1/\eps \log \alpha)$ update time per point where $\alpha$ is the ratio of the largest and smallest distance and is assumed to be known in advance. For general $k$, they gave a $(6+\eps)$-approximation using O($k/\eps \log \alpha)$ points and O($k^2/\eps \log \alpha)$ points. \\
\textbf{Problems in the sliding window model}.
Many problems have been studied in the literature~\cite{braverman-ostrovsky-2010,braverman-2016}.
Several algorithms have been proposed for the diameter problem~\cite{feigenbaum-2004,chan-Sadjad-2006,cohen-2016}. Chan and Sadjad~\cite{chan-Sadjad-2006} give a $(1+\epsilon)$-approximation using $O((1/\epsilon)^{(d+1)/2}\log \alpha / \epsilon)$ points. For higher dimensions, Cohen-Addad et al.~\cite{cohen-2016} give a $(3+\epsilon)$-approximation using $O(1/\epsilon \log \alpha)$ points and $O(1/\epsilon \log \alpha)$ update time per point.








\subsection{Detailed Comparison with~\cite{cohen-2016}}

For the $1$-center problem, when the value of $\alpha$ is known, our algorithm is a modification of the diameter algorithm of Cohen-Addad et al.~\cite{cohen-2016}. In addition to the center point returned by the algorithm of Cohen Addad, we maintain a second center point, such that all alive points are contained in two small balls centered at these two points. This way we obtain a better approximation factor. For the case of unknown $\alpha$, the number of our solutions is bounded, thereby removing the assumption that the value of $\alpha$ is known in advance. We show that it is sufficient to maintain O($1/\epsilon \log \alpha$) solutions.



A coreset is a small portion of the data, such that running a clustering algorithm on it, generates a clustering which is approximately optimal for the whole data.
For the $k$-center problem, we find a coreset of size $O(k)$ which guarantees a $(c+2\sqrt{3}+\epsilon)$-approximation by using a $c$-approximation for the coreset.
We use the observation that no three points within the Euclidean unit ball have all pairwise distances of more than $\sqrt{3}$.

\subsection{Our Contribution}


For the $1$-center problem, we obtain a $(3+\epsilon)$-approximation that works without knowing the parameter $\alpha$ in advance by adding a carefully chosen point. The parameter $\alpha$ is the ratio of largest and smallest possible distance between the points. Our algorithm maintains O($1/\eps \log \alpha$) points and O($1/\eps \log \alpha$) update time per point. We also remove the assumption that $\alpha$ is known in advance. Because $\alpha$ is changed when the sliding window moves, finding a proper estimate of $\alpha$ is difficult. Therefore it is important to remove the assumption for implementing an algorithm in the streaming model.
Our idea is general enough to be adapted to the algorithms of Cohen-Addad et al. for the metric diameter problem and the metric $k$-center problem.

In the static model, finding a $2$-approximation for the $1$-center problem is known and easy, however, in the streaming model finding a $\textit{feasible}$ radius is difficult because we do not know all the points. Therefore our result is non-trivial.

For the $k$-center problem, our algorithm finds a coreset of size $O(k)$ using O($k/\eps \log \alpha$) points and O($k \log k + k/\eps \log \alpha$) update time per point. 
Our algorithm gives a $(c+2\sqrt{3} + \epsilon)$-approximation for the Euclidean $k$-center problem by using any given $c$-approximation for the coreset where $c$ is a positive real number.
By using the exact algorithm~\cite{agarwal-procopiuc-2002} for the coreset, our algorithm gives a $(1+2\sqrt{3} + \eps)$-approximation  ($\approx 4.465$) using O$(k^{k^{1-(1/d)}}$) time per point.
By using the $2$-approximation~\cite{feder-greene-1988} for the coreset, our algorithm gives a $(2+2\sqrt{3} + \eps)$-approximation ($\approx 5.465$) using $O(k\log k)$ time per point. Our two algorithms for the $k$-center problem are an improvement on the approximation factor of $(6+\eps)$ by Cohen-Addad et al.~\cite{cohen-2016} with the same space complexity.

For low dimensional Euclidean space, better approximation is available.
Our algorithm finds a coreset of size $O(k)$ maintaining O($k2^{c_d t} \log(\alpha)/\epsilon$) points and O$(k \log k + k2^{c_d d t} \log(\alpha)/\epsilon$) update time per point where $t$ is an trade-off parameter, which is a positive integer, and $c_d$ is the doubling constant. 
We give a $(c+2\sqrt{3}(\frac{1}{2})^t+\epsilon)$-approximation by using any given $c$-approximation for the coreset. We can get a $(1+\epsilon)$-approximation by using this result.

\section{Preliminaries}
\label{section:preliminaries}

Let $P$ be a set of $n$ points in $d$-dimensional Euclidean space. In the sliding window model, the points in $P$ arrive one by one, and are allowed to be examined only once. The points in $P$ are labeled in order of their arrival. That is, $p_i$ is a point in $P$ that arrives in the $i$-th step. We denote by $P_N$ a subset of points in $P$ that are the last $N$ points, that is, $P_N = \{ p_{n-N+1}, p_{n-N+2}, \ldots, p_n \}$.
We call a point $p$ \textit{alive} if $p \in P_N$.
Let $\textit{Time}(p)$ be the index of $p$ in the insertion order. ($\textit{Time}(p_i) = i$.) Let $\textit{Dia}(P)$ be the diameter of the points set $P$. Let $\textit{CP}(P)$ be the closest pair of the points set $P$. Let $\alpha$ be the ratio between the diameter and the minimum non-zero distance between any two points in $P_N$.
Let $|xy|$ be the Euclidean distance between $x$ and $y$. Let $B(c,r)$ denote a ball of radius $r$ centered at $c$. Let $r^*$ be the radius of optimal solution.
\section{The 1-Center Problem}
\label{section:1-center}

\begin{algorithm}
\caption{\textit{Diameter}($\gamma$) }\label{alg:diameter}
\begin{multicols}{2}
\begin{algorithmic}[1]

\State $c_{old},q,r,b \gets$ first point of the stream
\State $c_{new} \gets$ null
\For {all elements $p$ of the stream}
\If {$c_{old}$ is deleted}
\If {$c_{new} \neq $null and $c_{old} = q$}
\State $c_{old} \gets r$; $b \gets c_{new}$;
\State $c_{new} \gets $ null; 
\EndIf
\If {$c_{new} \neq $null and $c_{old} \neq q$}
\State $b \gets c_{old}$; $c_{old} \gets q$;
\State $c_{new} \gets $ null;
\EndIf
\If {$c_{new} = $null}
\State $b \gets c_{old}$; $c_{old} \gets r$;
\EndIf
\EndIf
\State INSERT($p$)
\State $r \gets p$
\EndFor

\Procedure{INSERT}{p}
\If {$c_{new} =$ null}

\If {$|pr| > \gamma$}
\State $c_{old}$,$q \gets r$; $c_{new} \gets p$;
\ElsIf {$|pc_{old}| > \gamma$}
\State $q \gets r$; $c_{new} \gets p$;
\EndIf
\Else
\If {$|pr| > \gamma$}
\State $c_{old}$,$q \gets r$; $c_{new} \gets p$;

\ElsIf {$|pc_{new}| > \gamma$}
\State $c_{old} \gets c_{new}$; $q \gets r$; $c_{new} \gets p$;
\ElsIf {$|pq| > \gamma$}
\If {$c_{old} \neq q$}
\State $c_{old} \gets q$; $q \gets r$; $c_{new} \gets p$;
\EndIf
\EndIf
\EndIf

\EndProcedure
\end{algorithmic}
\end{multicols}
\end{algorithm}

First we give a $(6+\epsilon)$-approximation by using the algorithm for the diameter problem. The details of this can be found in the Appendix A.

Now we explain our $(3+\epsilon)$-approximation algorithm. Our idea in this section is to carefully adapt Algorithm 1 of Cohen-Addad et al.~\cite{cohen-2016}, originally designed for the diameter problem. To improve readability, we sketch their algorithm and its properties to explain our modifications.
Their algorithm consists of two parts: a fixed parameter algorithm and a way to maintain parameters. We only explain the first part.

For a given estimate $\gamma$, their algorithm maintains four points $c_{old}$, $c_{new}$, $q$, and $r$.
Their algorithm returns either two points $c_{old}$ and $c_{new}$ ($\textit{Time}(c_{old}) < \textit{Time}(c_{new})$) such that $|c_{old}c_{new}| \ge \gamma$ or a point $c_{old}$. If their algorithm returns one point, then they show that $|Dia(P_N)| \le 3 \gamma$ in this case.

They show that the following two invariants are satisfied.
The invariants are:
\begin{invariant}
\label{inv:1}
If $c_{new} = null$ then the following statements hold: \\
a) For any alive points $x,y$ with $\textit{Time}(x), \textit{Time}(y) \le \textit{Time}(c_{old})$, we have $|xy| \le 2\gamma$.	\\
b) For any point $x$ with $\textit{Time}(x) > \textit{Time}(c_{old})$, we have $|xc_{old}| \le \gamma$.
\end{invariant}

\begin{invariant}
\label{inv:2}
If $c_{new} \neq null$, then the following statements hold:	\\
a) $|c_{old}c_{new}| > \gamma$.	\\
b) For any point $x$ with $\textit{Time}(c_{old}) < \textit{Time}(x) < \textit{Time}(c_{new})$, we have $|xc_{old}| \le \gamma$.	\\
c) For any point $x$ with $\textit{Time}(c_{new}) < \textit{Time}(x)$, we have $|xc_{new}| \le \gamma$.	\\
d) If $c_{old} \neq q$ then for any point $x$ with $\textit{Time}(q) < \textit{Time}(x)$, we have $|xq| \le \gamma$.
\end{invariant}

Now we explain our algorithm. Our algorithm consists of two parts: a fixed parameter algorithm and a way to maintain parameters.

Given an estimate $\gamma$, we use their fixed parameter algorithm with maintaining  a new point. 
We call our algorithm $Diameter(\gamma)$.
We maintain a \textit{bridge} point and following an invariant show the property of the point. The main idea is that two balls $B(c_{old},\gamma)$ and $B(b,\gamma)$ contain all alive point if our algorithm returns a point $c_{old}$ where $b$ is a bridge point. See Figure~\ref{fig:onecenter} (b).
This invariant is a variant of the Invariant 1 and is given implicitly in the proofs of their Lemma 3 and Lemma 4~\cite{cohen-2016}. The proof of this Invariant can be found in the Appendix B.

\begin{lemma}[Invariant 3]
\label{inv:3}
If $c_{new} = null$ then the following statement holds:	\\
There exists a \textit{bridge} point $b$ such that $|xb| \le \gamma$ for any alive point $x$ with $\textit{Time}(x) \le \textit{Time}(c_{old})$.
\end{lemma}

By Invariant 1 and 3, we can find a solution. The proof of this Lemma can be found in the Appendix B.

\begin{lemma}
\label{small_ball}
If $c_{new} = null$ then $B(m,1.5\gamma)$ contains $P_N$ where $m = \frac{c_{old}+b}{2}$.
\end{lemma}

Now we explain the way to maintain estimates. First we explain the case when the number of estimate is unbounded and then we explain the case when the number of estimate is bounded.

We maintain estimates, which is an exponential sequence to the base of ($1+\epsilon$), such that any value between the distance of the closest pair and the distance of the diameter is ($1+\epsilon$) approximated. For each power of ($1+\epsilon$), we run $\textit{Diameter}(\gamma)$.

\begin{figure}[ht]
	\centering
	\includegraphics[width=0.6\textwidth]{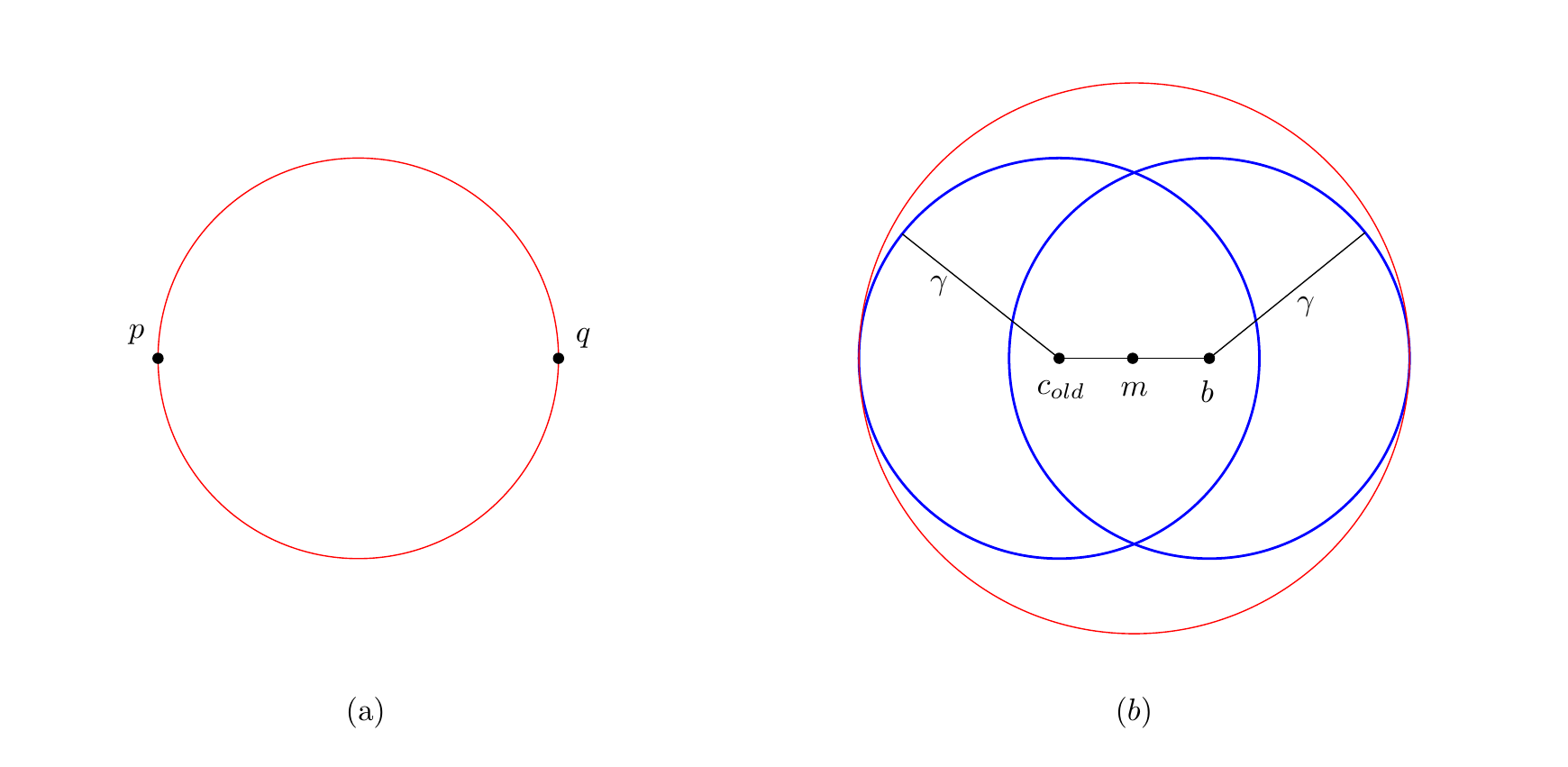}
	\caption{(a) $p$ and $q$ are diameter and $|pq| = 2 r^*$. Therefore $3|pq| \le 6 r^*$ (b) $B(m,1.5\gamma)$ contains $P_N$}
	\label{fig:onecenter}
\end{figure}

Now we explain the way to bound the number of estimates.
Let $\Upsilon$ be an estimates set containing $(1+\eps)^i $ for all integer $i$.
We modify some solutions of our algorithm such that the new solutions also satisfy all invariants and maintain the specific points.
Let $\gamma_L$ ($\gamma_U$) be an estimate in $\Upsilon$ such that, for any estimate $\gamma  \le \gamma_L$ ($\gamma \ge \gamma_U$),  $Diameter(\gamma)$ and $Diameter(\gamma_L)$ ($Diameter(\gamma_U)$) maintain the same points and the same solution.
By this property, we know all the solutions in $\Upsilon$ by just maintaining estimates between $\gamma_L$ and $\gamma_U$.

To find proper $\gamma_L$ and $\gamma_U$, we use the following witnesses. The proof of this Lemma can be found in the Appendix B.

\begin{lemma}
	If $\gamma \le |p_{n-1}p_n|$, then $\textit{Diameter}(\gamma)$ maintains $c_{old},r,q \gets p_{n-1}$ and $c_{new} \gets p_n$. It returns $c_{old}$ and $c_{new}$ as a solution.
	If $\gamma > |\textit{dia}(P_N)|$, then we can maintain $c_{old}$, $r$, $q 
\gets p_{n-1}$, $c_{new} \gets null$ , and $b \gets p_{n-1}$. We return $c_{old}$ as a solution. All the invariants hold for the solution.
\end{lemma}

We set the largest estimate satisfying $\gamma < |p_{n-1}p_n|$ as $\gamma_L$. Let $\gamma_{e}$ be the length of our solution. Because $\gamma_e \ge |Dia(P_N)|/3$ by Invariant 1, we set the smallest estimate satisfying $\gamma \ge 3\gamma_{e}$ as $\gamma_U$.

When a constant number of points are inserted, we maintain all estimates between $\gamma_L$ and $\gamma_U$ by using evidences for each direction.
Then we update $\gamma_L$ and $\gamma_U$. We first update the direction decreasing the number of estimates if it is possible.
Because $\gamma_L$ is larger than $|CP(P_N)|$ and $\gamma_U$ is smaller than $3|\textit{Dia}(P_N)|$, the number of estimates is $O(\log_{1+\eps}\alpha) = O(\frac{\log \alpha}{\log (1+\eps)}) = O(  \log (\alpha)/\eps)$.

Among estimates, we choose the smallest estimate $\gamma$ that $\textit{Diameter}(\gamma)$ returns a point $c_{old}$.
Note that $\textit{Diameter}(\frac{\gamma}{1+\epsilon})$ returns two points and its means $\frac{\gamma}{1+\epsilon} \le |\textit{Dia}(P_N)|$. Since $|\textit{Dia}(P_N)| \le 2\rst$, $\gamma \le 2(1+\eps)\rst$. Therefore our solution $B(m,1,5\gamma)$ guarantees $(3+3\eps)$-approximation. See Figure~\ref{fig:onecenter} (b).

Note that $\textit{Diameter}(\gamma)$ takes a constant time to update.

Therefore we obtain the following Theorem.

\begin{theorem}
\label{thm:1-center}
	Given a set of points $P$ with a window of size $N$ and a value $\epsilon > 0$, our algorithm guarantees a $(3+\eps)$-approximation to the Euclidean $1$-center problem in the sliding window model maintaining O$(\log(\alpha)/\epsilon)$ points and requiring update time per point in arbitrary dimensions $d$.
\end{theorem}
\section{The k-Center Problem}
\label{section:k-center}

Our algorithm is similar to Algorithm 2 of Cohen-Addad et al.~\cite{cohen-2016}. The main difference is that we implicitly maintain O($k$) balls of radius $\sqrt{3}r$ by observing a property of the Euclidean ball.
Our algorithm maintains a coreset and does not compute a solution for each update. When a query is given, our algorithm computes a solution from the coreset. 
Our algorithm consists of two parts: a fixed parameter algorithm and a way to maintain parameters.

We call our algorithm \textit{4kCoreSet}$(r)$. For a given an estimate optimal radius $r$, our algorithm gives a coreset $C$ of at most $4k$ points such that, for any alive point $x$ and its nearest point $y$ in $C$, $|xy| \le 2\sqrt{3}r$.
If our algorithm maintains $C$ satisfying the condition, we call the coreset \textit{feasible}. Otherwise we call the coreset \textit{infeasible}.
If $r \ge \rst$ our algorithm maintains a \textit{feasible} coreset. Otherwise our algorithm may maintain a feasible coreset or an infeasible coreset.
If $\rst < r \le (1+\eps) \rst$, then we can compute an $(c+2\sqrt{3} + \eps)$-approximation by computing $c$-approximation of $k$-center problem for the O($k$) points.
For small $k$, we compute an $(1+2\sqrt{3} + \eps)$-approximation solution by computing an optimal solution for the points.
For large $k$, we compute an $(2+2\sqrt{3} + \eps)$-approximation solution by computing a $2$-approximation solution for the points.
We maintain estimates and one of them satisfies the condition.


\begin{figure}[ht]
	\centering
	\includegraphics[width=0.6\textwidth]{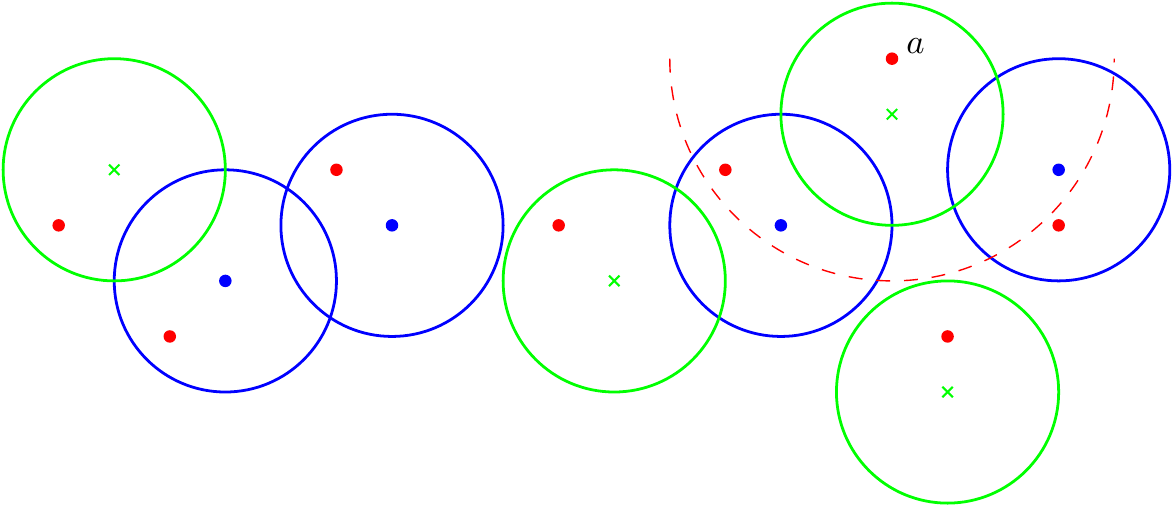}
	\caption{An example for the $2$-center problem . The centers of balls are active center points and raidus is $\sqrt{3}r$. The centers of blue balls are in $A$ and the centers of green balls are not in $A$. The red points are the representative points. The dashed red ball is $B(a,2\sqrt{3}r)$ and it contains the green ball whose representative point is $a$.}
	\label{fig:4kdiscrete}
\end{figure}

A high level description of our algorithm is as follows. We implicitly maintain at most $4k$ balls with radius $\sqrt{3}r$ such that the balls contain all points in $P_N$. We maintain a representative point per each ball and return the representative points as our coreset. In Figure~\ref{fig:4kdiscrete}, we give an example of our solution for the $2$-center problem.

We explain our algorithm more precisely.
We maintain a set $A$ of at most $2k$ active center points. For each active center point $a$, we maintain a representative point $R(a)$ within radius $\sqrt{3}r$. We maintain a set $R$ for the representative points.
When new point $p$ is arrived, we first remove points in $R$ and $A$ that are deleted. Note that the representative point of a deleted active center point can be in $R$.
Then we choose all active center points from $A$ such that the distance from $p$ is at most $\sqrt{3}r$, and if there are points we update corresponding representative points to $p$. Otherwise we add $p$ to $A$.
If $|A| = 2k+1$, we remove oldest active center point $a_{old}$. In this case, we set this coreset as \textit{infeasible} until $a_{old}$ is deleted and delete all representative points equal or older than the removed active center point.
We do this process to bound the number of points we maintain.
To decide whether our coreset is \textit{infeasible} or not, we maintain the \textit{feasible} time $\textit{FT}$ and the current time $\textit{CT}$.
We set $\textit{FT}=\textit{Time}(c_{old}+N)$ and $\textit{CT} = \textit{Time}(p)$.
If $\textit{FT} > \textit{CT}$, then this solution is \textit{infeasible}.

If our coreset is \textit{feasible}, then all points in $P_N$ are in the union of at most $4k$ balls we maintain implicitly.
Moreover, each ball contains a representative point in $R$.
Let $\mathcal{B} = \{ B_1, B_2, \ldots, B_k \}$ be a set of congruent balls such that $R$ is contained in the union of balls in $\mathcal{B}$. We enlarge the balls by $2\sqrt{3}r$, then the enlarged balls contain all points in $P_N$. 

\begin{algorithm}
\caption{$4k$CoreSet($r$) }\label{alg:k-center}
\begin{multicols}{2}
\begin{algorithmic}[1]
\State $A, R \gets \emptyset$
\State $\textit{CT}, \textit{FT} \gets 0$
\For {all element $p$ of the stream}
\State $\textit{CT} \gets \textit{Time}(p)$
\If	{$x \in R$ is deleted}
\State $R \gets R \setminus \{x\}$
\EndIf
\If	{$a \in A$ is deleted}
\State DeleteActive(a)
\EndIf
\State Insert(p)
\EndFor

\Procedure{DeleteActive}{a}
\If	{$a$ is not deleted}
\State $\textit{FT} \gets \textit{Time}(a) + N$
\For {$x \in R$}
\If {$\textit{Time}(x) \le \textit{Time}(a)$}
\State {$R \gets R \setminus \{  x \}$}
\EndIf
\EndFor
\EndIf
\State $A \gets A \setminus \{ a \}$
\EndProcedure

\Procedure{Insert}{p}
\State $D \gets \{ a \in A | |pa| \le \sqrt{3}r \}$
\If {$D = \emptyset$}
\If {$|A| = 2k$}
\State $a_{old} \gets$ oldest point in $A$ 
\State DeleteActive($a_{old}$)
\EndIf
\State $A \gets A \cup \{ p \}$
\State $R(p) \gets p$
\State $R \gets R \cup \{ R(p) \}$
\Else
\For {all $a\in D$}
\State $R(a) \gets p$
\EndFor
\EndIf
\EndProcedure
\end{algorithmic}
\end{multicols}
\end{algorithm}

We start our analysis by giving the space bound.

\begin{lemma}
\label{lem:store_k}
	$4k$CoreSet($r$) uses O($k$) space and $O(k)$ update time per point.
\end{lemma}
\begin{proof}
	The number of active center points in $A$ is at most $2k$ and the number of representative points corresponding the active center points is at most $2k$.
	What remains to be shown is that the number of representative points whose active center point is not in $A$ is at most $2k$. This result come from Lemma 7 of Cohen-Addad~\cite{cohen-2016}, but for the completeness we explain it.
	
	Let $a_i$ be $i$th active center point inserted in $A$.
	Note that after an active center point $a$ is removed from $A$, $R(a)$ is not updated. Therefore, $\textit{Time}(R(a_j)) < \textit{Time}(a_{j+2k})$ for all $j$.	
	Assume we currently have $A = \{a_i,\cdots ,a_{i-1+m} \}$ where $m \le 2k$. If $a_{i-1}$ is deleted, then $R(a_{j-1-2k})$ is also deleted for all $j \le i$.	
	 If $a_{i-1}$ is not deleted, then we removed it from $A$ and we also removed all representative points older than $a_{i-1}$ by line 13 to line 15 of $4k$CoreSet($r$). Therefore, we removed $R(a_j)$ from $R$ for all $j \le i-1 -2k$. 

By the above reason, $|A| \le 2k$ and $|R| \le 4k$, and the space bounds holds.

Now we bound time to update. Removing points in the main procedure takes $O(k)$. The procedure \textit{INSERT} takes $O(k)$.
\end{proof}

Now we show that our algorithm return a \textit{feasible} solution when $r \ge \rst$. We need the following technical lemma to show it.

\begin{lemma}
\label{lem:2point_in_ball}
	Let $B$ be a unit ball in $\bbbr^d$ for $d \ge 1$. There are no three point $p$, $q$, and $r$ such that the points are contained in $B$ and all of their pairwise distances are larger than $\sqrt{3}$.
\end{lemma}
\begin{proof}
	In order to derive a contradiction, we assume that there exist such points $p$, $q$, and $r$. Then we choose the plane passing the points. Let $B'$ be the 2-dimensional ball intersection of $B$ and the plane. Note that the radius of $B'$ is at most $1$ and the convex hull of the points lies in $B'$. Then we can move three points such that their pairwise distances are $\sqrt{3}$ and those points are contained in the convex hull of the origin points. Let those points be $p'$, $q'$, and $r'$. Because $p'$, $q'$, and $r'$ are on the boundary of $B'$ when $B'$ is a unit ball, at least one of the original points lies outside of $B'$. See the Figure~\ref{fig:triangle}. This is a contradiction.
\end{proof}

By Lemma~\ref{lem:2point_in_ball}, the following lemma hold.

\begin{figure}[ht]
	\centering
	\includegraphics[width=0.25\textwidth]{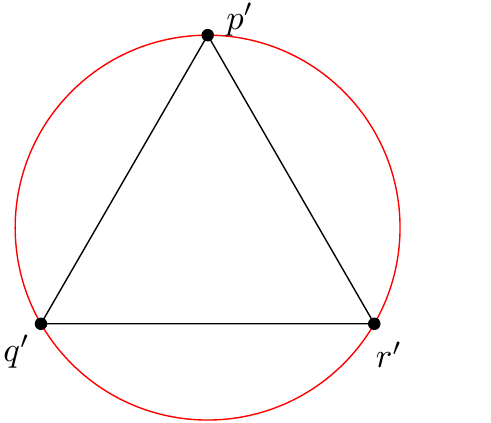}
	\caption{$p'$, $q'$, and $r'$ are on the boundary of $B'$.}
	\label{fig:triangle}
\end{figure}

\begin{lemma}
\label{lem:2ball_cotain}
	Let $B$ be a unit ball in $\bbbr^d$ for $d \ge 1$. For any points $c_1$ and $c_2$ in $B$ and $dist(c_1,c_2) > \sqrt{3}$. Then $B \subset B(c_1,\sqrt{3}) \cup B(c_2,\sqrt{3})$.  
\end{lemma}

Now we show that our algorithm maintains a \textit{feasible} coreset when $r \ge \rst$.

\begin{lemma}
\label{lem:feasible}
	If $r \ge \rst$, then $4k$CoreSet($r$) maintains a \textit{feasible} coreset.	
\end{lemma}
\begin{proof}
	In $4k$CoreSet($r$), we maintain at most $2k$ active center points in $A$. Assume we currently have $A = \{ a_i, a_{i+1}, \cdots, a_{i-1+m}  \}$ where $m \le 2k$.
	
	If $a_{i-1}\in P_N$, then this lemma holds.
	If $a_{i-1}\notin P_N$, then we will show it is impossible when $r \ge \rst$. In order to derive a contradiction, we assume that $a_{i-1}$ is alive.
	Since $a_{i-1}$ is alive, we have at least $2k+1$ active center points whose pairwise distance is larger than $\sqrt{3}r$. By Lemma~\ref{lem:2ball_cotain} and the distances between the points, at most two points of the points lie in an optimal ball. Therefore the number of the points are at most $2k$, but we have more than $2k$, it is a contradiction.
\end{proof}

Now we explain the way to maintain estimates. We use a similar way as in Section~\ref{section:1-center}.
We maintain estimates, which is an exponential sequence to the base of ($1+\epsilon$), such that any value between the distance of the closest pair and the distance of the diameter is ($1+\epsilon$) approximated. For each power of ($1+\epsilon$), we run \textit{4kCoreSet}$(r)$. We choose the smallest $r$ that returns a feasible solution.

Let $r_L$ ($r_U$) be an estimate in $\Upsilon$ such that, for any estimate $r  \le r_L$ ($r \ge r_U$), \textit{4kCoreSet}$(r)$ and \textit{4kCoreSet}$(r_L)$ (\textit{4kCoreSet}$(r_U)$) maintain the same coreset.
To bound the number of estimates, we use the same idea as in Section~\ref{section:1-center}. Let $\Upsilon$ be an estimates set containing $(1+\eps)^i $ for all integer $i$. Let $r_L$ ($r_U$) be an estimate in $\Upsilon$ such that for any estimate $r  \le r$ ($r \ge r$) \textit{4kCoreSet}$(r_L)$ (\textit{4kCoreSet}$(r_U)$) and \textit{4kCoreSet}$(r)$ maintain the same coreset.
We set $r_L < \frac{|CP(P_{2k+1})|}{\sqrt{3}}$ to the witness.
For any estimate $r \le r_L$, balls $B(p_j,\sqrt{3}r)$ are disjoint where $n-2k\le j \le n$,
therefore $R = A = \{p_{n-2k+1}, p_{n-2k+1}, \ldots , p_n \}$ is an \textit{infeasible} coreset until $\textit{Time}(p_{n-2k})$.
We set $r_U > \frac{2|\textit{dia}(P_N)|}{\sqrt{3}}$ to the witness. Because $B(p_n,\sqrt{3} r_U)$ contains all points in $P_N$, we maintain $A=R= \{ p_n \}$ as the coreset for $r \ge r_U$. We set $r_U = \frac{6(1+\epsilon)\gamma_{e}}{\sqrt{3}}$ where $\gamma_{e}$ is the estimate of the solution of our $1$-center problem.

Combining these lemmas and ideas we have :

\begin{theorem}
\label{thm:main}
Given a set of points $P$ with a window of size $N$ and a value $\eps > 0$, our algorithm maintains at most $4k$ points $R$ such that for any point $x \in P_N$ there is a point $y \in R$ such that $dist(x,y) \le (2\sqrt{3}+\eps)\rst$. Our algorithm  maintains O$(k \log(\alpha)/\epsilon)$ points and requires O$(k\log k + k\log (\alpha)/\epsilon)$ update time per point.
To compute a $(c+2\sqrt{3}+\epsilon)$-approximate solution to the Euclidean $k$-center problem we need a $c$-approximation solution of the $k$-center problem algorithm for O$(k)$ points.
\end{theorem}
\begin{proof}
By Lemma~\ref{lem:feasible}, $4k$CoreSet($r$) returns a \textit{feasible} coreset when $r \ge \rst$. Now we show that our algorithm is a $(c+2\sqrt{3}+\epsilon)$-approximation when $\rst < r \le (1+\frac{\eps}{2\sqrt{3}})\rst$. Our algorithm gives a set $R$ of at most $4k$ points. For any alive point $p$, there exists a point $x \in R$ such that $dist(p,x) \le 2\sqrt{3}r < 2\sqrt{3}(1+\frac{\eps}{2\sqrt{3}})\rst = (2\sqrt{3}+\eps)\rst$ by lemma~\ref{lem:feasible}. Let $C = \{ c_1, c_2, \ldots, c_k \}$ be $c$-approximation solution for $k$-center problem for $R$. For any alive point $x$, the distance from its closest center in $C$ is at most $(c+ 2\sqrt{3})r$.
Therefore, the union of balls $B_i = B(c_i, (c+ 2\sqrt{3})r)$ $1 \le i \le k$ contains all points in $P_N$ and it is a $(c+2\sqrt{3}+\epsilon)$-approximation.

For finding $r_L$, we use a $2$-approximation of the $2k$-center problem for $P_{2k+1}$. Note that the closest pair are in an optimal ball and other optimal balls contain exactly one point. Therefore, the radius is between  $|CP(P_{2k+1})|/2$ and $|CP(P_{2k+1})|$. It takes O($k \log k$) time by the algorithm of Feder and Greene~\cite{feder-greene-1988}.

The memory usage of the algorithm consists of O($k$) per instance of \textit{4kCoreSet}$(r)$ and $\log_{1+\eps}\alpha = \frac{\ln \alpha}{\ln (1+\eps)} \le \frac{2}{\eps}  \ln \alpha$ estimates.

$4k$CoreSet($r$) takes $O(k)$ update time per point and $O(\frac{2}{\eps}  \log \alpha)$ estimates. To bound the number of estimate it takes $O(k \log k)$ time per point.
\end{proof}
\section{The k-Center Problem in Low Dimension}
\label{section:k-center_low}
In this section, we have the following theorem. The details of this section can be found in the Appendix C. 

\begin{theorem}
\label{thm:lowD_doubling}
Given a set of points $P$ with a window of size $N$ and a value $\eps > 0$, our algorithm maintains $O(2^{(c_dd)t})$ points $R$ such that for any point $x \in P_N$ there is a point $y \in R$ such that $dist(x,y) \le (\sqrt{3}/2^{t-1}+\eps)\rst$. Our algorithm  maintains O$(k2^{c_d d t} \log(\alpha)/\epsilon)$ points and requires O$(k \log k + k2^{c_d d t} \log(\alpha)/\epsilon$) update time per point.
To compute a $(c+2\sqrt{3}(\frac{1}{2})^t+\epsilon)$-approximation to the Euclidean $k$-center problem we need a $c$-approximation of the $k$-center problem algorithm for the coreset.
\end{theorem}

\renewcommand{\abstractname}{\ackname}

\begin{abstract}
  We thank all anonymous referees for their valuable input.
We thank Rolf Klein, David Kübel, Elmar Langetepe and Barbara Schwarzwald for their discussions.
\end{abstract}

\footnotesize

\bibliographystyle{abbrv}
\bibliography{EKCenterSW}

\begin{thebibliography}{10}

\bibitem{agarwal-procopiuc-2002}
P.~K. Agarwal and C.~M. Procopiuc.
\newblock Exact and approximation algorithms for clustering.
\newblock {\em Algorithmica}, (33):201–226, 2002.

\bibitem{agarwal-shara-2010}
P.~K. Agarwal and R.~Sharathkumar.
\newblock Streaming algorithms for extent problems in high dimensions.
\newblock {\em In Proc. 21st ACM-SIAM Sympos. Discrete Algorithms.},
  3(12):1481–1489, 2010.

\bibitem{assouad-1983}
P.~Assouad.
\newblock Plongements lipschitziens dans $r^n$.
\newblock In {\em Societe mathematique de France}, pages 429--448, 1983.

\bibitem{bern-eppstein-1996}
M.~Bern and D.~Eppstein.
\newblock Approximation algorithms for geometric problems.
\newblock In {\em PWS Publishing Co}, 1996.

\bibitem{braverman-2016}
V.~Braverman, H.~Lang, K.~Levin, and M.~Monemizadeh.
\newblock Clustering problems on sliding windows.
\newblock In {\em 27th Annual ACM-SIAM Symposium on Discrete Algorithms}, page
  1374–1390, 2016.

\bibitem{braverman-ostrovsky-2010}
V.~Braverman and R.~Ostrovsky.
\newblock Effective computations on sliding windows.
\newblock {\em SIAM J. Comput}, 39(6):2113--2131, 2010.

\bibitem{chan-1999}
T.~M. Chan.
\newblock More planar two-center algorithms.
\newblock {\em Comput. Geom}, 3(13):189–198, 1999.

\bibitem{chan-pathak-2014}
T.~M. Chan and V.~Pathak.
\newblock Streaming and dynamic algorithms for minimum enclosing balls in high
  dimensions.
\newblock {\em Comput. Geom.}, 2(47):240–247, 2014.

\bibitem{chan-Sadjad-2006}
T.~M. Chan and V.~Sadjad.
\newblock Geometric optimization problems over sliding windows.
\newblock {\em Int. J. Comput. Geometry Appl}, 16(2-3):145–158, 2006.

\bibitem{chazelle-matousek-2015}
B.~Chazelle and J.~Matousek.
\newblock On linear-time deterministic algorithms for optimization problems in
  fixed dimension.
\newblock {\em J. Algorithms}, 3(21):579–597, 1996.

\bibitem{cohen-2016}
V.~Cohen-Addad, C.~Schwiegelshohn, and C.~Sohler.
\newblock Diameter and k-center in sliding windows.
\newblock In {\em In 43rd International Colloquium on Automata, Languages, and
  Programming}, pages 1--12, 2016.

\bibitem{feder-greene-1988}
D.~Feder and D.~H. Greene.
\newblock Optimal algorithms for approximate clustering.
\newblock In {\em Proc. 20th Annu. ACM Symp. Theory of Computing}, pages
  434--444, 1988.

\bibitem{feigenbaum-2004}
J.~Feigenbaum, S.~Kannan, and J.~Zhang.
\newblock Computing diameter in the streaming and sliding-window models.
\newblock {\em Algorithmica}, 41(1):25–41, 2004.

\bibitem{fowler-paterson-1981}
R.~J. Fowler, M.~S. Paterson, and S.~L. Tanimoto.
\newblock Optimal packing and covering in the plane are np-complete.
\newblock {\em Inform. Process. Lett.}, 3(12):133--137, 1981.

\bibitem{garey-johnson-1979}
M.~R. Garey and D.~S. Johnson.
\newblock Computers and intractability: A guide to the theory of
  np-completeness.
\newblock In {\em W.H. Freeman, New York}, 1979.

\bibitem{guha-2009}
S.~Guha.
\newblock Tight results for clustering and summarizing data streams.
\newblock In {\em In Proc. 12th Internat. Conf. Database Theory,}, pages
  268--275, 2009.

\bibitem{heinonen-2001}
J.~Heinonen.
\newblock Lectures on analysis on metric spaces.
\newblock In {\em Universitext. Springer New York}, 2001.

\bibitem{kim-ahn-2015}
S.-S. Kim and H.-K. Ahn.
\newblock An improved data stream algorithm for clustering.
\newblock {\em Comput. Geom.}, 9(48):635–645, 2015.

\bibitem{mccutchen-khuller-2008}
R.~M. McCutchen and S.~Khuller.
\newblock Streaming algorithms for k-center clustering with outliers and with
  anonymity pages.
\newblock In {\em In Proc. 11th Internat. Workshop Approx. Algorithms,}, pages
  165--178, 2008.

\bibitem{megiddo-1990}
M.~Megiddo.
\newblock On the complexity of some geometric problems in unbounded dimension.
\newblock {\em J. Symbolic Comput}, (10):327--334, 1990.

\bibitem{megiddo-supowit-1984}
M.~Megiddo and K.~J. Supowit.
\newblock On the complexity of some common geometric location problems.
\newblock {\em SIAM J. Comp}, 1(13):182--196, 1984.

\bibitem{verger-2005}
J.~L. Verger-Gaugry.
\newblock Covering a ball with smaller equal balls in $r^n$.
\newblock {\em Discrete $\&$ Computational Geometry}, 1(33):143--155, 2005.

\bibitem{zarrabi-2008}
H.~Zarrabi-Zadeh.
\newblock Core-preserving algorithms.
\newblock In {\em Proc. 20th Canadian Conf. Computational Geometry}, page
  159–162, 2008.

\bibitem{zarrabi-chan-2006}
H.~Zarrabi-Zadeh and T.~M. Chan.
\newblock A simple streaming algorithm for minimum enclosing balls.
\newblock In {\em IIn Proc. 18th Canad. Conf. Computat. Geom}, pages 139--142,
  2006.

\end{thebibliography}
%
%
%
%

\newpage
\appendix
\section*{Appendices}
\renewcommand{\thesubsection}{\Alph{subsection}}

\subsection{The $(6+\epsilon)$-approximation for $1$-center problem}
First we give a $(6+\epsilon)$-approximation by using the algorithm for the diameter problem.
The metric diameter problem is to find two points of maximum distance among a set of points lying in some metric space.
The algorithm of Cohen-Addad et al.~\cite{cohen-2016} returns two points $p$ and $q$ such that their distance is at least $\frac{(1-\epsilon)}{3}|\textit{Dia}(P_N)|$. We choose the last point as the center and $3|pq|$ as the radius. We give an example that gives $(6+\epsilon)$-approximation. See Figure~\ref{fig:onecenter} (a). In the figure, the length of the diameter is $2r^*$.

\subsection{Proofs of lemmas in Section~\ref{section:1-center}}
\begin{lemma}[Invariant 3]
\label{inv:3}
If $c_{new} = null$ then the following statement holds:	\\
There exists a \textit{bridge} point $b$ such that $|xb| \le \gamma$ for any alive point $x$ with $\textit{Time}(x) \le \textit{Time}(c_{old})$.
\end{lemma}
\begin{proof}
We consider the situation when Invariant 3 holds and a new point $p$ is inserted. We will show that Invariant 3 is satisfied after the update.
Let $c'_{old}$ ($c'_{new}$) be the point $c_{old}$ ($c_{new}$) before $p$ is inserted. Let $r'$, $q'$ and $b'$ be defined similar way.

There are 3 cases when $c_{new}=null$ and $b \neq b'$. \\
Case 1 : $c'_{new} = null$ and $c'_{old}$ is deleted. 		\\
Case 2 : $c'_{new} \neq null$ and $c'_{old} = q$, and $c'_{old}$ is deleted.	\\
Case 3 : $c'_{new} \neq null$ and $c'_{old} \neq q$, and $c'_{old}$ is deleted.

We explain how to update the bridge point $b$ for each case. Note that for any case $c'_{old}$ is deleted.\\
\textbf{Case 1)} In this case, their algorithm sets $c_{old} = r'$. We need to show that for any point $x$ with $\textit{Time}(c'_{old}) < \textit{Time}(x) \le \textit{Time}(r')= \textit{Time}(c_{old})$, there is a point $b$ such that $|xb| \le \gamma$. By Invariant 1.b), for any point $x$ with $\textit{Time}(c'_{old}) < \textit{Time}(x) \le \textit{Time}(r')$, we have $|c'_{old}x| \le \gamma$. We set $b = c'_{old}$.	\\
\textbf{Case 2)} In this case, their algorithm sets $c_{old} = r'$. We need to show that for any point $x$ with $\textit{Time}(c'_{old}) < \textit{Time}(x) \le \textit{Time}(r')= \textit{Time}(c_{old})$, there is a point $b$ such that $|xb| \le \gamma$. Since $c'_{old} = q$ and $\textit{Time}(q) = \textit{Time}(c'_{new})-1$, $c'_{new}$ is the oldest alive point. And, by Invariant 2.c), for any point $x$ with $\textit{Time}(c'_{new}) \le \textit{Time}(x) \le \textit{Time}(r')$, we have $|xc'_{new}| \le \gamma$. We set $b = c'_{new}$.	\\
\textbf{Case 3)} In this case, their algorithm sets $c_{old} = q'$. We need to show that for any point $x$ with $\textit{Time}(c'_{old}) < \textit{Time}(x) \le \textit{Time}(q') < \textit{Time}(c'_{new})$, there is a point $b$ such that $|xb| \le \gamma$. By Invariant 2.b), for any point $x$ with $\textit{Time}(c'_{old}) < \textit{Time}(x) = \textit{Time}(c_{old}) < \textit{Time}(c'_{new})$, we have $|xc'_{old}| \le \gamma$. We set $b = c'_{old}$.
\end{proof}

\begin{lemma}
\label{small_ball}
If $c_{new} = null$ then $B(m,1.5\gamma)$ contains $P_N$ where $m = \frac{c_{old}+b}{2}$.
\end{lemma}
\begin{proof}
By Invariant 3, we have a \textit{bridge} point $b$, for any alive point $x$ that is older than $c_{old}$, satisfying $|xb| \le \gamma$. By Invariant 1, $|xc_{old}| \le \gamma$ for any $x$ with $Time(x) > Time(c_{old})$. $|bc_{old}| \le \gamma$ by Invariant 3. We choose the mid point $m$ of $b$ and $c_{old}$. Then for any alive point $x$ $|xm| \le 1.5 \gamma$. See Figure~\ref{fig:onecenter} (b).
\end{proof}

\begin{lemma}
	If $\gamma \le |p_{n-1}p_n|$, then $\textit{Diameter}(\gamma)$ maintains $c_{old},r,q \gets p_{n-1}$ and $c_{new} \gets p_n$. It returns $c_{old}$ and $c_{new}$ as a solution.
	If $\gamma > |\textit{dia}(P_N)|$, then we can maintain $c_{old}$, $r$, $q 
\gets p_{n-1}$, $c_{new} \gets null$ , and $b \gets p_{n-1}$. We return $c_{old}$ as a solution. All the invariants hold for the solution.
\end{lemma}
\begin{proof}
	First we prove the case when $\gamma \le |p_{n-1}p_n|$.
	After updating $p_{n-1}$, we set $r$ to $p_{n-1}$ by line 14 of Algorithm 1. When $p_n$ is inserted our algorithm call \textit{Insert}$(p)$. In the procedure, whether $c_{new}$ is null or not, $|pr| > \gamma$. So we set $c_{old}$ and $q$ to $r$ and $c_{new}$ to $p$. Since $c_{new}$ is not null, we do not maintain $b$. In line 14, we set $r$ to $p_n$.

When $\gamma > |\textit{dia}(P_N)|$, we set $c_{old}$, $r$ and $q$ to $p_{n-1}$, $c_{new}$ to null, and $b$ to $p_{n-1}$. Note that Invariant 1 and Invariant 3 hold because $\gamma > |\textit{dia}(P_N)|$.
\end{proof}

\subsection{The k-Center Problem in Low Dimension}


In this section, we explain a better approximation when the dimension of Euclidean space is low. The idea is similar to Section~\ref{section:k-center}. Our algorithm consists of two parts: a fixed parameter algorithm and a way to maintain parameters. We focus on the fixed paprmeter algorithm.
We implicitly maintain small balls that contain all alive point. 
For a given radius of balls, we will bound the number of balls and approximation factor of our algorithm by using a property of doubling metric space.

The doubling dimension~\cite{assouad-1983,heinonen-2001} of a metric space is the smallest $D$ such that every ball of radius $r$ is covered by $2^D$ balls of radius most $r/2$. It is well-known that a point set in $d$-dimensional Euclidean metric has doubling dimension $\Theta (d)$. In~\cite{verger-2005}, Jean-Louis give a constant $c_d$ ( equation (4) in Theorem 1.2) such that $D \le c_dd$.




\begin{table}[ht]
\caption{Relation between the approximation factor and the number of balls} 
\centering 
\begin{tabular}{c c c } 
\hline\hline 
$t$     & APX 			& balls \\ [0.5ex] 
\hline 
1 		& $2.733+\eps$	& O($2^{c_d d}$) 		\\
3 		& $1.434+\eps$	& O($2^{3c_d d}$) 		\\
5 		& $1.109+\eps$	& O($2^{5c_d d}$) 		\\
8 		& $1.014+\eps$	& O($2^{8c_d d}$) 		\\
10 		& $1.004+\eps$	& O($2^{10c_d d}$) 		\\ 
$\log \sqrt{3} - \log \epsilon + 1$ & $1+\epsilon$ & O($2^{(\log \sqrt{3} - \log \epsilon + 1)c_d d}$)												\\[1ex] 
\hline 
\end{tabular}
\label{table:nonlin} 
\end{table}


We will use the property of doubling metric space in the following way.
A unit ball is contained in $2^{c_dd}$ balls with radius $1/2$. We can use this idea recursively. A ball with radius $1/2$ is contained in $2^{c_dd}$ balls with radius $(1/2)^2$. Therefore $2^{(c_dd)t}$ balls with radius $(1/2)^t$ contain a unit ball.
By lemma~\ref{lem:2ball_cotain}, a ball with radius $(1/2)^t$ is contained in the union of two balls with radius $\sqrt{3}(1/2)^t$.


The basic idea of our algorithm is similar to our algorithm in Section~\ref{section:k-center}. The difference is that we maintain a set $A$ of at most $2^{(c_dd)t}$ active center points and the distance between an active center point and its representative point is at most $\sqrt{3}(1/2)^tr$.
If $\rst < r \le (1+\eps) \rst$, then our algorithm returns \textit{feasible} solution and it guarantees a $(c+2\sqrt{3}(\frac{1}{2})^t+\epsilon)$-approximation by using $c$-approximation for the coreset. We maintain estimates with the similar way in Section~\ref{section:1-center}.
To bound the number of estimates, we set $r_L < \frac{|CP(P_{2^{(c_dd)t}+1})|}{\sqrt{3}}$ and $r_U \ge \frac{2|dia(P_N)|}{\sqrt{3}}$.

\begin{theorem}
\label{thm:lowD_doubling}
Given a set of points $P$ with a window of size $N$ and a value $\eps > 0$, our algorithm maintains $O(2^{(c_dd)t})$ points $R$ such that for any point $x \in P_N$ there is a point $y \in R$ such that $dist(x,y) \le (\sqrt{3}/2^{t-1}+\eps)\rst$. Our algorithm  maintains O$(k2^{c_d d t} \log(\alpha)/\epsilon)$ points and requires O$(k \log k + k2^{c_d d t} \log(\alpha)/\epsilon$) update time per point.
To compute a $(c+2\sqrt{3}(\frac{1}{2})^t+\epsilon)$-approximation to the Euclidean $k$-center problem we need a $c$-approximation of the $k$-center problem algorithm for the coreset.
\end{theorem}

To get $(1+\epsilon)$-approximation, inequality $(1/2)^t \sqrt{3} \le \epsilon/2$ holds and we get $t \ge \log \sqrt{3} - \log \epsilon + 1$ by modifying the inequality.




\end{document}